\documentclass[copyright]{eptcs} 
\providecommand{\event}{ICE 2013} 

\usepackage[usenames]{color}
\usepackage{makeidx}
\usepackage{verbatim}
\usepackage{fancyhdr}
\usepackage{fancybox}
\usepackage{stmaryrd}
\usepackage[reqno]{amsmath}
\usepackage{amsthm}
\usepackage{amssymb}
\usepackage{amsfonts}
\usepackage{xypic}
\usepackage {indentfirst}
\usepackage{graphicx}
\usepackage[hang]{subfigure}
\usepackage{cite}

\usepackage{proof}

\usepackage{bbding}

\newtheorem{theorem}{Theorem}

\newtheorem{lemma}[theorem]{Lemma}

\theoremstyle{definition}
\newtheorem{definition}[theorem]{Definition}

\theoremstyle{remark}

\newcommand{\para}{\,|\,}
\newcommand{\fosub}[2]{\{#1/#2\} }
\newcommand{\hosub}[2]{\{#1/#2\} }

\newcommand{\vect}[1]{{#1_1,#1_2,...,#1_n} }
\newcommand{\ve}[1]{\widetilde{#1}}
\newcommand{\size}[1]{|#1|}

\newcommand{\st}[1]{{\xrightarrow{#1}} }
\newcommand{\wt}[1]{{\stackrel{#1}{\Longrightarrow}} }

\newcommand{\nsepv}[1]{\vspace{0mm}}

\newcommand{\DEF}{\stackrel{def}{=}} 

\newcommand{\lrangle}[1]{\langle #1 \rangle} 

\title{On Context Bisimulation for Parameterized Higher-order Processes}
\author{Xian Xu \thanks{The author acknowledges the support by the National Nature Science Foundation of China (61202023,61261130589,61173048), and the ANR project 12IS02001 PACE.}
\institute{Department of Computer Science and Technology}
\institute{East China University of Science and Technology, Shanghai, China P.R. (200237)}
\email{xuxian@ecust.edu.cn, xuxian2004@gmail.com}
}
\def\titlerunning{Context Bisimulation for Parameterized HO Processes}
\def\authorrunning{Xian Xu}
\begin{document}
\maketitle

\begin{abstract}
This paper studies context bisimulation for higher-order processes, in the presence of parameterization (viz. abstraction). We show that the extension of higher-order processes with process parameterization retains the characterization of context bisimulation by a much simpler form of bisimulation called normal bisimulation (viz. they are coincident), in which universal quantifiers are eliminated; whereas it is not the same with name parameterization. These results clarify further the bisimulation theory of higher-order processes, and also shed light on the essential distinction between the two kinds of parameterization. 
\vspace*{.1cm}

\end{abstract}


%
%
%
%



\vspace*{-.7cm}
\section{Introduction}
Higher-order processes differ from first-order (name-passing) ones in that they can transmit themselves (i.e. integral programs) in communication. This mechanism of process-passing provides an alternative yet essentially distinct way from name-passing to achieve mobility. That distinction lies in several aspects, among which the behavioral theory is of pivotal importance. As a basis of behavioral theory, bisimulation theory has been studied since the very early work on higher-order processes. The most well-know (and probably standard) bisimulation is context bisimulation \cite{San92}. It was proposed to improve on the previous forms of bisimulation, including (applicative) higher-order bisimulation \cite{Tho90}\cite{Tho93}, by considering the sent process and residual process at the same time. It then continued to draw attention \cite{San94}\cite{VD98}\cite{JR05}\cite{Cao06}\cite{LPSS08}\cite{SKS11}. 

\subsection*{Related work and motivation}
Context bisimulation, in its original form, is not so convenient in that, for example, it involves universal quantifiers when comparing output (as well as input), i.e. in the output clause of the definition of context bisimulation one has to check the matching of the output action with respect to all possible contexts. Here matching means the output action of one process can be simulated with an output by the other and the resulting derivatives are bisimilar again. That is a particularly heavy burden. This situation is improved by the so-called normal bisimulation, which characterizes context bisimulation (i.e. coincident with it) but requires only the checking with some special context (see an example below) \cite{San92}\cite{San94}\cite{Cao06}. 

A common methodology for establishing such a characterization in \cite{San92}\cite{San94}\cite{Cao06} consists of two main ingredients: (1) Showing a factorization theorem using triggers. 
The factorization theorem provides a mechanism of relocating a (sub)process to a new place, and setting up a pointer to the new place for potential use. Such a pointer is represented by a trigger. 
(2) Showing the coincidence of context and normal bisimulation with the help of an intermediate bisimulation equivalence called triggered bisimulation, which is defined on a subclass of higher-order processes communicating only triggers. This design of normal bisimulation, particularly the special context used in it, is guided by the factorization theorem as described above.


We exemplify below the simplification of context bisimulation by normal bisimulation in basic higher-order processes \cite{San94}, equipped with the basic operators including input prefix ($a(X).P$ in which $X$ is a bound variable), output prefix ($\overline{a}A.P$), parallel composition ($P\para Q$) and restriction ($(c)P$ in which $c$ is a bound name). 
Notice that $E[X]$ denotes a process with the free occurrence of $X$ (i.e. not bound by an input prefix $a(X).P$), $E[A]$ denotes substituting $A$ for all free occurrences of $X$ in $E[X]$, output action $(\ve{c})\overline{a}A$ means sending over (port) name $a$ a process $A$ containing a set $\ve{c}$ of bound names, and the replication operation $!m.A$ is a derivable one in a higher-order setting \cite{Tho93}. We use $\approx$ for context bisimilarity (i.e. the equivalence induced by context bisimulation), and $\cong$ for normal bisimilarity (complete definitions are given in Definitions \ref{context-bisimulation},\ref{normal-bisi-bigd}). 
First of all, the factorization theorem is stated as below ($m$ is fresh, i.e. it does not occur in $E[A]$).
\[
E[A] \approx (m)(E[Tr_m]\para !m.A), \qquad \mbox{ where the trigger }Tr_m\equiv \overline{m}.0
\] Intuitively, using factorization theorem one can extract from $E[A]$ the process $A$ that might cause difference to its behavior. 
For instance, 
\[
A\para Q \approx (m)((\overline{m}.0\para Q) \para !m.A)
\]

Now suppose $P$ and $Q$ are basic higher-order processes, and are bisimilar with respect to either bisimulation equivalence (i.e. context bisimulation or normal bisimulation). We focus on the output clauses of context bisimulation and normal bisimulation, since the input case is similar. Note $fn(E[X])$ returns the free names of $E[X]$ (i.e. not bound by restriction), and for simplicity we do not claim the existence of $\ve{d},B,Q'$ below. 
\begin{itemize}
\item The output clause of context bisimulation is \\
 If $P\st{(\ve{c})\overline{a}A} P'$, then $Q\wt{(\ve{d})\overline{a}B}Q'$ s.t. for every $E[X]$ with $fn(E[X])\cap (\ve{c}\cup \ve{d}) = \emptyset$, $(\ve{c})(P'\para E[A])\approx (\ve{d})(Q'\para E[B])$.

\item The output clause of normal bisimulation is \\
 If $P\st{(\ve{c})\overline{a}A} P'$, then $Q\wt{(\ve{d})\overline{a}B}Q'$ s.t. for some fresh $m$, $(\ve{c})(P'\para !m.A) \cong (\ve{d})(Q'\para !m.B)$.
\end{itemize}
The intuition behind the simplification is that using the factorization theorem one can extract process $A$ (respectively $B$) in $E[A]$ (respectively $E[B]$), so in context bisimulation one obtains the following, due to the congruence property of $\approx$,
\begin{center}
\xymatrix{
R_1\DEF (\ve{c})(P'\para (m)(E[Tr_m]\para !m.A)) \ar@{.}[d]^{\approx} \ar@{.}[r]^{\approx} & (\ve{d})(Q'\para (m)(E[Tr_m]\para !m.B))\DEF R_2 \ar@{.}[d]^{\approx} \\
(\ve{c})(P'\para E[A]) \ar@{.}[r]^{\approx} & (\ve{d})(Q'\para E[B])
}
\end{center}
where each dotted line connects two processes that are related by context bisimilarity ($\approx$). 
Then by the congruence property of $\approx$, one can eliminate the common part in processes $R_1$ and $R_2$ (i.e. $E[Tr_m]$ and restriction on $m$), and thus arrive at the form of the output clause in normal bisimulation.


The above has explained how the normal bisimulation works in basic higher-order processes \cite{San94}\cite{Cao06}. This paper concentrates on the extension of normal bisimulation, including its relevant technique, to higher-order processes with parameterization (also known as abstraction, akin to that in lambda-calculus; we postpone some examples until before stating the contribution), which still lacks full discussion. 

Although parameterization is considered in \cite{San92}, the definition of higher-order processes in \cite{San92} includes both name-passing and process-passing, and the coincidence between context bisimulation and normal bisimulation is studied in such a mixed language. So one does not know clearly whether the normal bisimulation can be extended to a pure higher-order setting with parameterization. 
To the best of our knowledge, it has not been discussed thoroughly whether such characterization, together with related technique such as trigges, is still applicable in a pure higher-order calculus extended with the two kinds of parameterization (on names and on processes) respectively. 
In general, it is still not clear whether the characterization result can be extended to an enriched higher-order setting. 
Yet the question whether the context bisimulation has such a convenient characterization of normal bisimulation in a parameterized setting is significant in broader research, as well as in improving the bisimulation theory itself (i.e. no universal checking is demanded). 
It would be interesting to know the answer to this question also because it is already known that parameterization can strictly enhance the expressiveness power of mere process-passing \cite{LPSS10}; thus potential difference in the behavior theory of a more expressive calculus can be revealed. 
Relevantly in the study of expressiveness, a simpler yet equivalent form of context bisimulation would probably render much easier the proof of soundness (i.e. two processes are equivalent only if their encodings are equivalent), which requires the establishing of context bisimulation in a handy way. For instance, when comparing name-passing and processing \cite{San92a}\cite{Xu12}, the soundness demands that two name-passing processes are equivalent (e.g. early bisimilar) only if their translations to process-passing processes are context bisimilar. In this case, establishing normal bisimulation instead of context bisimulation would be much more convenient.

In this paper, we explicitly examine the normal characterization of context bisimulation in pure higher-order processes with parameterization; this time the language is not mixed, that is we study each parameterization separately, so that the essential difference can be revealed. We use $\Pi$ to denote the basic (strict) higher-order pi-calculus, with the elementary operators (higher-order input and output, parallel composition, restriction) and without any first-order fragment. Notation $\Pi^D_n$ (respectively $\Pi^d_n$) stands for the calculus extending $\Pi$ with process parameterization (respectively name parameterization) of arity $n$ ($n\in \mathbb{N}$ and $\mathbb{N}$ is the set of natural numbers); $\Pi^D$ (respectively $\Pi^d$) is the union of $\Pi^D_n$ (respectively $\Pi^d_n$). For example, the process $P_1$ below is a $\Pi$ process (it means outputting $X$ on name $a$ and continuing as null); $P_2$ is a $\Pi^D_1$ process; $P_3$ is a $\Pi^d_1$ process. The parameterization operation is denoted by (leftmost) $\lrangle{\cdot}$. In the case of $P_2$, the $X$ is a parameterized variable that can be instantiated with an application (rightmost $\lrangle{\cdot}$), for instance $P_2\lrangle{A}$ results in the process $\overline{a}A.0$. We will formally define the calculi in section \ref{def-languages}.
\[
\begin{array}{lll}
P_1 \DEF \overline{a}X.0 \qquad&
P_2 \DEF \lrangle{X}\overline{a}X.0 \qquad&
P_3 \DEF \lrangle{x}\overline{x}X.0
\end{array}
\]
More often than not, process (respectively name) parameterization is understood as certain abstraction on processes (respectively names) in the literature. So parameterization and abstraction may be used interchangeably. 


\subsection*{Contribution}
The main contribution of this paper can be summarized as below. 
\begin{itemize}
\item In the calculus $\Pi^D_n$ ($\Pi^D$ as well), we have normal bisimulation and it indeed characterizes context bisimulation. We prove the coincidence by re-exploiting the available method from \cite{San92}\cite{San94}. Moreover, a technical novelty here  is that the proof is given in a more direct way, rather than first resorting to an intermediate bisimulation called triggered bisimulation by restricting to a sub-class of processes as described above. 

\item The approach of normal bisimulation (i.e. simple characterization of context bisimulation with normal bisimulation) cannot be extended to the calculus $\Pi^d_n$ ($\Pi^d$ either). That is, one cannot reuse the explicit technique of `normal' bisimulation from \cite{San92}\cite{San94}. We provide a counterexample and some detailed discussion. 
\end{itemize}
Since bisimulation theory often serves as the basis (or tool) for more advanced studies, these results provide a reference point for  related work on higher-order calculi. They also shed light on the expressiveness of parameterization, because 
the contrast between the two results above reveals an essential separation between parameterization on process and names. Furthermore, such distinction also stresses the intrinsic complexity of name-handling (in a strictly higher-order setting without name-passing). So it will be interesting to look for some other kind of simpler characterization of context bisimulation in presence of name parameterization. 

\subsection*{Organization}
The paper is organized as follows.  
In section \ref{def-languages}, we define the calculi $\Pi,\Pi^D_n,\Pi^d_n$, and the standard form of context bisimulation.  Section \ref{bigD-charac} proves that there indeed exists the normal bisimulation that characterizes context bisimulation in $\Pi^D_n$ ($n\in \mathbb{N}$). On the other hand, $\Pi^d_n$ does not have such characterization in general, which we show in section \ref{smalld-charac}. In section \ref{conclusion}, we provide the conclusions and some further discussions.





\section{Calculi}\label{def-languages}
In this section, we define the calculus $\Pi$, and its variants $\Pi^D_n, \Pi^d_n$ within which the context bisimulation is examined. 


\subsection{Calculus $\Pi$}
The $\Pi$ processes, denoted by uppercase letters ($A,B,E,F,P,Q,R,T...$) and their variant forms (e.g.  $T'$), are defined by the following grammar ($\Pi_{seg}$ is used to provide convenience for defining the variants of $\Pi$). Lowercase letters represent  channel names, whereas $X,Y,Z$ stand for process variables.
\[
\begin{array}{l}
T ::= \Pi_{seg}  \\
\Pi_{seg} ::= 0 \;\Big{|}\; X \;\Big{|}\;  u(X).T \;\Big{|}\; \overline{u}T'.T \;\Big{|}\; T\para T' \;\Big{|}\; (c)T  
\end{array}
\]
The operators are: prefix ($u(X).T, \overline{u}T'.T$), parallel composition ($T\para T'$), restriction ($(c)T$) in which $c$ is bound (or local); they have their standard meaning, and parallel composition has the least precedence. 
As usual some convenient notations are: $a$ for $a(X).0$; $\overline{a}$ for $\overline{a}0.0$; $\overline{m}A$ for $\overline{m}A.0$; $\tau.P$ for $(a)(a.P\para \overline{a})$; sometimes $\overline{a}[A].T$ for $\overline{a}A.T$ (for the sake of  clarity); $\ve{\cdot}$ for a finite sequence of some items (e.g. names, processes), and $\ve{c}\ve{d}$ for the concatenation of $\ve{c}$ and $\ve{d}$.  
By standard definition, $fn(\ve{T})$, $bn(\ve{T})$, $n(\ve{T})$; $fv(\ve{T})$, $bv(\ve{T})$, $v(\ve{T})$ respectively denote free names, bound names, names; free variables, bound variables and variables in $\ve{T}$. Closed processes contain no free variables, and are studied by default. 
A fresh name or variable is one that does not occur in the processes under consideration.
Name substitution $T\fosub{\ve{n}}{\ve{m}}$ and higher-order substitution $T\hosub{\ve{A}}{\ve{X}}$ are defined structurally in the standard way. 
$E[\ve{X}]$ denotes $E$ with variables $\ve{X}$, and $E[\ve{A}]$ stands for $E\hosub{\ve{A}}{\ve{X}}$. 
We work up-to $\alpha$-conversion and always assume no capture.  
We use the following version of replication as a derived operator \cite{Tho93}\cite{LPSS08}:  
$
  !\phi.P \DEF (c)(Q_c \para \overline{c}Q_c),\;
  Q_c \DEF c(X).(\phi.(X\para P) \para \overline{c}X)
$, where $\phi$ is a prefix. 

The operational semantics is given in Figure \ref{lts-Pi}. Symmetric rules are omitted. 
The rules are mostly self-explanatory. For example, in higher-order input ($a(A)$), the received process $A$ becomes part of the receiving environment through a substitution; in higher-order output ($(\ve{c})\overline{a}A$), the process $A$ is sent with a set of local names for prospective use in further communication. 
In the first rule of the second row, we slightly abuse the notation, i.e. $fn(A){-}\{\ve{c},a\}$ means the free names of $A$ minus the names in $\ve{c}$ and $a$, which excludes the possibility of $d{=}a$. 
Symbols $\alpha,\beta,\lambda,...$ denote actions, whose subject (e.g. $a$ in action $a(A)$) indicates the channel name on which it happens.  Operations $fn(), bn(), n()$ can be similarly defined on actions. 
$\wt{}$ is the reflexive transitive closure of the internal transition ($\st{\tau}$), and $\wt{\lambda}$ is $\wt{}\xrightarrow{\lambda}\wt{}$.  $\wt{\hat{\lambda}}$ is $\wt{}$ when $\lambda$ is $\tau$ and $\wt{\lambda}$ otherwise. $\st{\tau}_k$ means $k$ consecutive $\tau$'s. 
$P\wt{}\cdot \mathcal{R}\, Q$ means $P\wt{} Q'$ for some $Q'$ and $Q' \,\mathcal{R}\, Q$ (i.e. $(Q',Q)\in \mathcal{R}$), where $\mathcal{R}$ is a binary relation.  
We say relation $\mathcal{R}$ is closed under (variable) substitution if $(E\hosub{A}{X},F\hosub{A}{X})\in \mathcal{R}$ for any $A,X$ whenever $(E,F)\in \mathcal{R}$, in which $E,F$ (possibly) have free occurrence of variable $X$. 
A process diverges if it can perform an infinite $\tau$ sequence. 

\vspace*{-.55cm}
\begin{figure}[htbp]
\[\begin{array}{llll}
\frac{\displaystyle }{\displaystyle a(X).T\st{a(A)} T\hosub{A}{X}} \; &    
 \frac{}{\displaystyle \overline{a}A.T\st{\overline{a}A} T} \; &       
\frac{\displaystyle T\st{\lambda} T'}{\displaystyle (c)T\st{\lambda} (c)T'}c\not\in n(\lambda) \;\;&
 \frac{\displaystyle T_1\st{a(A)} T_1', T_2\st{(\ve{c})\overline{a}[A]} T_2'}{\displaystyle T_1\para T_2 \st{\tau}(\ve{c})(T_1'\,|\,T_2')}
\end{array}
\]

\[\begin{array}{ll}
\frac{\displaystyle T\st{(\ve{c})\overline{a}[A]} T'}{\displaystyle (d)T\st{(d\ve{c})\overline{a}[A]} T'}d \in fn(A){-}\{\ve{c},a\} \quad & 
\frac{\displaystyle T\st{\lambda} T'}{\displaystyle T\para T_1\st{\lambda} T'\para T_1} bn(\lambda)\cap fn(T_1)=\emptyset 
\end{array}
\]

\vspace*{-.3cm}
\caption{Semantics of $\Pi$}\label{lts-Pi}
\end{figure}

\vspace*{-.8cm}
\subsection{Calculi $\Pi^D_n$ and $\Pi^d_n$}
Parameterization extends $\Pi$ with the syntax and semantics in Figure \ref{PiDdn}. 
$\lrangle{\vect{U}} T$ is a parameterization where $\vect{U}$ are the pairwise distinct formal parameters to be instantiated by the application $T\lrangle{\vect{K}}$ where the parameters are instantiated by concrete objects $\vect{K}$. 
Application binds tighter than prefixes and restriction. 
In the rule of Figure \ref{PiDdn}, $\size{\ve{U}}{=}\size{\ve{K}}$ requires the sequence of parameters and the sequence of instantiating objects should be equally sized. It also expresses that the parameterized process can do an action only after the application happens. 
Calculus $\Pi^D_n$, which has process parameterization (or higher-order abstraction),  is defined by setting $\ve{U},\ve{K}$ to be  $\ve{X},\ve{T'}$ respectively. 
Calculus $\Pi^d_n$, which has name parameterization (or first-order abstraction),  is defined by setting $\ve{U},\ve{K}$ to be $\ve{x},\ve{u}$ respectively. 
For convenience, names (ranged over by $u,v,w$) are divided into two disjoint subsets: name constants (ranged over by $a,b,c,...,m,n$); name variables (ranged over by $x,y,z$).
\begin{figure}[htbp]

\[\begin{array}{l}
T::= \Pi_{seg}  \;\Big{|}\; \lrangle{\vect{U}} T \;\Big{|}\;  T\lrangle{\vect{K}}  \\ 
\frac{\displaystyle T\hosub{\ve{K}}{\ve{U}} \stackrel{\lambda}{\longrightarrow } T'}{\displaystyle F\lrangle{\ve{K}}\stackrel{\lambda}{\longrightarrow } T'}\;\;\mbox{ if } F\DEF \lrangle{\ve{U}}T \;\;(\size{\ve{U}}=\size{\ve{K}}=n)
\end{array}
\]

\caption{$\Pi$ with parameterization}\label{PiDdn}
\end{figure}
Process expressions (or terms) of the form $\lrangle{\ve{X}} P$ or $\lrangle{\ve{x}} P$, in which $\ve{X}$ and $\ve{x}$ are not empty, are \emph{parameterized processes}. Terms without outmost parameterization are \emph{non-parameterized processes}, or simply processes. We mainly focus on (closed) processes. Only free  variables can be effectively parameterized (i.e. it does not make sense to parameterize a bound variable); they become \emph{bound} after parameterization. In the syntax, redundant parameterizations, for example $\lrangle{X_1,X_2}P$ in which $X_2\notin fv(P)$, are allowed. A $\Pi^D_n$ ($n\geq 1$) process is therefore definable in $\Pi^D_{n+1}$ (similar for $\Pi^d_n$). In the case of $\lrangle{\ve{X}}P$, it can be coded up by setting an additional fresh dummy variable $Y$ (of no use) to obtain $\lrangle{\ve{X},Y}P$. Sometimes related notations are slightly abused if no confusion is caused.

Type systems for the processes of $\Pi^D_n$ and $\Pi^d_n$ can be routinely defined in a similar way to that in \cite{San92}. 
We do not present the type system and always assume type consistency, since such a type system is not important for the discussion in this paper. 
Without loss of generality, we stipulate that the processes of $\Pi^D_n$ and $\Pi^d_n$ are strictly abstraction-passing, i.e. all the transmitted objects are parameterized processes; accordingly, it is assumed that all the process variables have the type of abstractions, so an occurrence of a variable $X$ typically takes the form $X\lrangle{T'}$ in some context. This can be justified by two facts: firstly a non-parameterized process can be treated as a special case of parameterization; secondly to the aim of this paper, the characterization of context bisimulation in the case of non-parameterized processes has been examined in depth by Sangiorgi \cite{San94}. 

\vspace*{-.4cm}
\begin{figure}[htbp]
\centering
\[
\begin{array}{l}
T\equiv T', \mbox{ if they are $\alpha$-convertible to each other},\\
\quad\quad\quad\;\;\; i.e.,\; a(X).T\equiv a(Z).T\hosub{Z}{X}, (c)T\equiv (d)T\fosub{d}{c} \\[.5ex]
T\para 0 \equiv T, \, T\para (T'\para T'') \equiv (T\para T')\para T'', \, T \para T'  \equiv T'\para T \\[.5ex]
(c)(d)T \equiv (d)(c)T, \,(c)\lambda.T\equiv 0 \mbox{ whenever the subject of $\lambda$ is $c$} \\[.5ex]
(c)(T\para T') \equiv (c)T\para T' \mbox{ whenever } c\notin fn(T') \\[.5ex]
(\lrangle{\ve{U}}T)\lrangle{\ve{K}}\equiv T\hosub{\ve{K}}{\ve{U}}, \size{\ve{U}}=\size{\ve{K}}=n
\end{array}
\]
\vspace*{-.3cm}
\caption{Structural congruence}\label{struc-congru}
\vspace*{-.4cm}
\end{figure}
The semantics of parameterization renders it somewhat natural to deem application as some (extra) rule of structural congruence, 
denoted by $\equiv$, which is defined in Figure \ref{struc-congru} in a standard way \cite{Mil92}\cite{San94}. In addition to the standard algebraic laws (concerning parallel composition and restriction), the last rule of application is included. 
So the rule below can used in place of that in Figure \ref{PiDdn}.
\[\infer{T\st{\alpha} T' }{T\equiv T_1, T_1\st{\alpha} T_2, T_2\equiv T'}
\]

\subsection{Context Bisimulation}
The bisimulation equivalence of higher-order processes we intend to examine is the context bisimulation. The form of its definition is the same for the calculi defined above.
\begin{definition}[Context bisimulation]\label{context-bisimulation}
A symmetric binary relation $\mathcal{R}$ on closed processes is a context bisimulation, if
whenever $P\,\mathcal{R}\, Q$ the following properties hold: 
\begin{enumerate}
\item If $P \st{\alpha} P'$ and $\alpha$ is $\tau$ or $a(A)$, then $Q \wt{\hat{\alpha}} Q'$ for some $Q'$ and $P'\,\mathcal{R}\, Q'$;

\item If $P \st{(\ve{c})\overline{a}A} P'$ then $Q \wt{(\ve{d})\overline{a}B} Q'$ for some $\ve{d},B,Q'$, and for every process $E[X]$
 s.t. $\{\ve{c},\ve{d}\}\cap fn(E)=\emptyset$ it holds that
$(\ve{c})(E[A]\para P') \; \mathcal{R}\;  (\ve{d})(E[B]\para Q')$.
\end{enumerate}
Process $P$ is context bisimilar to $Q$, written $P\,\approx, Q$, if $P\,\mathcal{R}\, Q$ for some context bisimulation $\mathcal{R}$. Relation $\approx$ is called context bisimilarity. 
\end{definition}

Context bisimulation can be extended to general (open) processes in a standard way, i.e. suppose $fv(T,T')\subseteq \ve{X}$, then $T\approx T'$ if $T\hosub{\ve{A}}{\ve{X}} \approx T'\hosub{\ve{A}}{\ve{X}}$ for all closed $\ve{A}$.  
{Similarly the extension to parameterized processes is: $\lrangle{\ve{X}}T \approx \lrangle{\ve{X}}T'$ if $T\hosub{\ve{A}}{\ve{X}} \approx T'\hosub{\ve{A}}{\ve{X}}$ all closed $\ve{A}$. } 

Relation $\sim$ is the strong version of $\approx$. 
For clarity, notation $\approx_{\mathcal{L}}$ (resp. $\sim _{\mathcal{L}}$) indicates the context bisimilarity (resp. strong context bisimilarity) of calculus $\mathcal{L}$ ($\Pi,\Pi^D_n$, or $\Pi^d_n$); we simply use $\approx$ (resp. $\sim$) when it is clear from context.
It is well-known that context bisimilarity is \emph{an equivalence and a congruence} with respect to prefixing, parallel composition and restriction; see \cite{San92}\cite{San92a}\cite{San94}\cite{San96}. 
{
Notice $E[X]$ is different from the well-known concept of contexts, which neglect name capture (i.e. the case a free name falls into the scope of some restriction of the same name). That is $E[X]$ is sensitive to name capture and should use $\alpha$-conversion to avoid that. Otherwise, such two $\alpha$-convertible processes like $(m)\overline{a}[m.0].\overline{m}.b$ and $(n)\overline{a}[n.0].\overline{n}.b$ would be distinguishable by context bisimilarity using $E[X]\equiv (m)X$ as the receiving environment (the latter can produce a visible action on $b$ while the former does not), which is contradictory because $\alpha$-convertibility shall entail context bisimilarity. 
}




\section{Normal bisimulation in $\Pi^D_n$}\label{bigD-charac}
In this section, we show that in $\Pi^D_n$  we have normal bisimulation that characterizes context bisimulation. 
For convenience, we focus on $\Pi^D_1$; the result can be readily extended to $\Pi^D_n$. 
We first define normal bisimulation, and then prove the coincidence theorem.
Hereinafter $Tr_m\DEF \lrangle{Z}\overline{m}Z$ denotes a trigger with (trigger name) $m$.  

The form of normal bisimulation  and the proof schema stem from the result in \cite{San92}\cite{San94}\cite{Cao06}. However, as mentioned, we go beyond those works in several respects. Firstly, the processes under inspection here are purely higher-order without name-passing, and capable of parameterization on processes only. Secondly, although the schema for the proof of the characterization of context bisimulation using normal bisimulation exploits those works, the technical details are not the same. More specifically, different from the approaches in \cite{Cao06}, where the so-called index technique is essentially used to deal with actions,
and in \cite{San92}\cite{San94}, where an intermediate equivalence called triggered bisimulation is used and processes are first transformed into a subclass called triggered processes, we prove the coincidence of normal bisimulation and context bisimulation in a more \emph{direct} and easier way. 
 
\subsection{Definition of normal bisimulation}
\begin{definition}\label{normal-bisi-bigd} 
A symmetric binary relation $\mathcal{R}$ on closed processes of $\Pi^D_1$ is a normal bisimulation, if
whenever $P\,\mathcal{R}\, Q$ the following properties hold: 
\begin{enumerate}
\item If $P \st{\tau} P'$, then $Q \wt{} Q'$ for some $Q'$ s.t. $P'\,\mathcal{R}\, Q'$;

\item If $P \st{a(Tr_m)} P'$ and $Tr_m\equiv \lrangle{Z}\overline{m}Z$ ($m$ is fresh), then $Q \wt{a(Tr_m)} Q'$ for some $Q'$ s.t.  $P'\,\mathcal{R}\, Q'$;

\item If $P \st{(\ve{c})\overline{a}A} P'$ then $Q \wt{(\ve{d})\overline{a}B} Q'$ for some $\ve{d},B,Q'$, 
and it holds that ($m$ is fresh)
\[(\ve{c})(P'\para !m(Z).A\lrangle{Z}) \; \mathcal{R}\;  (\ve{d})(Q'\para  !m(Z).B\lrangle{Z})
\] 
which can be rephrased as $(\ve{c})(P'\para E'[A]) \; \mathcal{R}\;  (\ve{d})(Q'\para  E'[B])$
where the $E'[X] \DEF !m(Z).X\lrangle{Z}$ holding $A$ and $B$ is special, in contrast to the general requirement in context bisimulation. 
\end{enumerate}
Process $P$ is normal bisimilar to $Q$, written $P\,\cong, Q$, if $P\,\mathcal{R}\, Q$ for some normal bisimulation $\mathcal{R}$. Relation $\cong$ is called normal bisimilarity. 
\end{definition}


\noindent\textbf{\emph{Remark}} 
The strong version of $\cong$ is denoted by $\simeq$. 
Notation $\cong_{\mathcal{L}}$ (resp. $\simeq _{\mathcal{L}}$) indicates the normal bisimilarity (resp. strong normal bisimilarity) of calculus $\mathcal{L}$, if any; we simply use $\cong$ (resp. $\simeq$) when there is no confusion.
It can be shown in a standard way that normal bisimilarity is \emph{an equivalence and a congruence};  
see \cite{San92}\cite{San94} for a reference. 

The rest of this section is mainly devoted to the (technical) proof of the following theorem. 
\begin{theorem}\label{normal-characterization-bigd} 
In $\Pi^D_1$, normal bisimilarity coincides with context bisimilarity; that is $\cong \,=\, \approx$.  
\end{theorem}

\subsection{Normal bisimulation characterizes context bisimulation}
A key step in proving Theorem \ref{normal-characterization-bigd} is to prove the Factorization theorem (Theorem \ref{factor-bigd}). First we need some preparation, i.e. the two lemmas below. Intuitively, these two lemmas state some distributive laws concerning the replication $ !m(Z).A\lrangle{Z}$ ($m$ is fresh) that are useful in the proof of the Factorization theorem.

\begin{lemma}\label{premise-factor-0}
Suppose $E[X], E_1[X], E_2[X]$ belong to $\Pi^D_1$, and 
let $m\notin fn(E,E_1,E_2,A)$. 
\begin{enumerate}
\item[(1)] If $m\notin fn(\alpha)$, then 
\[
\begin{array}{l}
 (m)(\alpha.E[Tr_m] \para  !m(Z).A\lrangle{Z}) \approx \alpha.(m)(E[Tr_m] \para  !m(Z).A\lrangle{Z}) 
\end{array}
\]

\item[(2)] It holds for output prefix that
\[
\begin{array}{l}
 (m)(\overline{a}B_1.E_1[Tr_m] \para  !m(Z).A\lrangle{Z}) 
\approx  (m)(\overline{a}B_2.E_1[Tr_m] \para  !m(Z).A\lrangle{Z})
\end{array}
\] where $B_1\equiv E_2[Tr_m]$, $B_2\equiv (m)(E_2[Tr_m] \para  !m(Z).A\lrangle{Z})$.

\item[(3)] It holds for parallel composition that
\[
\begin{array}{ll}
& (m)(E_1[Tr_m]\para E_2[Tr_m] \para  !m(Z).A\lrangle{Z}) \\
\approx &  (m)(E_1[Tr_m] \para  !m(Z).A\lrangle{Z}) \para (m)(E_2[Tr_m] \para  !m(Z).A\lrangle{Z})
\end{array}
\]
\end{enumerate}
\end{lemma}
\begin{proof}
The proof is based on a standard argument on establishing bisimulations, where the scope of restriction concerning $m$ is the critical part. The details are thus omitted. 
\end{proof}

\begin{lemma}\label{premise-factor-1}
Suppose $E_1[X]$ belongs to $\Pi^D_1$, and let $m$ be fresh. It holds for every $A,B$ (of the type of abstraction) that 
\[B\lrangle{(m)(E_1[Tr_m] \para  !m(Z).A\lrangle{Z})} \approx (m)(B\lrangle{E_1[Tr_m]} \para  !m(Z).A\lrangle{Z}) 
\]
\end{lemma}
\begin{proof}
Suppose $B\equiv \lrangle{Y}T$, it amounts to prove 
\[
\begin{array}{ll}
 & T\hosub{T_1}{Y} \approx (m)(T\hosub{T_2}{Y} \para  !m(Z).A\lrangle{Z}) \\
\mbox{where } & T_1\equiv  (m)(E_1[Tr_m] \para  !m(Z).A\lrangle{Z}) \mbox{ and } T_2 \equiv E_1[Tr_m]
\end{array}
\]
This can be done by a simple induction on the structure of $T$ in a routine way. The details pertaining to establishing bisimulation during the induction would not raise major obstacle. 
\end{proof}

Now we are ready to show the  Factorization theorem. As mentioned, this theorem offers some method to relocate a subprocess, which might cause difference in behavior, and set up a reference to it with the help of a trigger, while maintaining the equivalence with respect to context bisimilarity. The proof of Theorem \ref{factor-bigd} is put in appendix \ref{appendix-proof-bigd-factor}.  
\begin{theorem}[Factorization]\label{factor-bigd} 
Given a $E[X]$ of $\Pi^D_1$, let $m$ be fresh and notice $Tr_m\equiv \lrangle{Z}\overline{m}Z$. It holds for every $A$ (of the type of abstraction) that 
 
\begin{itemize}
\item[(1)] if $E[Tr_m]$ is non-parameterized, i.e. not an abstraction, then 
 \[E[A] \approx (m)(E[Tr_m] \para  !m(Z).A\lrangle{Z}) 
\]

\item[(2)]  else if $E[Tr_m]\equiv \lrangle{Y_1}\cdots\lrangle{Y_k}E'$ for some $k\geq 1$ and non-parameterized $E'$, 
then 
\[
E[A] \approx \lrangle{Y_1}\cdots\lrangle{Y_k}((m)(E' \para  !m(Z).A\lrangle{Z})) 
\]
\end{itemize}

\end{theorem}

With the help of factorization theorem (Theorem \ref{factor-bigd}), below we give the proof of Theorem \ref{normal-characterization-bigd}. 
\begin{proof}[Proof of Theorem \ref{normal-characterization-bigd}]
The fact that $\approx$ implies $\cong$ barely needs argument, because the former demands more and the latter is actually a special case of it. 
So we focus on the other direction. To achieve this, we show the relation $\mathcal{R}\DEF \{ (P,Q) \,|\, P\cong Q \}$ (i.e. normal bisimilarity $\cong$) is a context bisimulation up-to $\approx$ \cite{San98}\cite{SW01a} (the definition is standard and thus omitted), by using mainly the factorization theorem. 

There are several cases to analyze in terms of the definition of context bisimulation. Notice we use weak transitions in the bisimulation, which is somewhat a standard variant of the corresponding bisimulation. Moreover we focus on the case when the first result (1) of theorem \ref{factor-bigd} applies, the case when the other applies can be handled in a similar (and easier) way. 
\begin{itemize}
\item Internal action. This case is trivial, because the clauses in context bisimulation and normal bisimulation are the same. 

\item Input. 
If $P \wt{a(A)} P'$, then we want to show that 
\begin{eqnarray}
Q \wt{a(A)} Q' \mbox{ for some $Q'$} \label{input-goal-1} \\
\mbox{ and } P' \approx\,\mathcal{R}\, \approx Q'  \label{input-goal-2} 
\end{eqnarray}
W.l.o.g. suppose $P'\equiv E[A]$ for some $E[X]$ (i.e. from $P\wt{a(X)}$ intuitively). So $P \wt{a(Tr_m)} E[Tr_m]$ for some fresh $m$.  
Since $P\cong Q$, we know
\[Q\wt{a(Tr_m)} F[Tr_m] \mbox{ for some } F 
\] and
\begin{equation}\label{mid-1}
E[Tr_m]\cong F[Tr_m]
\end{equation}
Thus 
\[Q\wt{a(A)} F[A]\DEF Q'
\] which fulfills (\ref{input-goal-1}). We know from (\ref{mid-1}) and the congruence properties of $\cong$ that
\[
(m)(E[Tr_m] \para  !m(Z).A\lrangle{Z}) \cong (m)(F[Tr_m] \para  !m(Z).A\lrangle{Z})
\] Now by factorization theorem (Theorem \ref{factor-bigd}), we have 
\[
E[A]\approx (m)(E[Tr_m] \para  !m(Z).A\lrangle{Z}) \,\mathcal{R}\, (m)(F[Tr_m] \para  !m(Z).A\lrangle{Z}) \approx F[A]
\] which arrives at (\ref{input-goal-2}).


\item Output.
 If $P \wt{(\ve{c})\overline{a}A} P'$, then we want to show that (notice $\{\ve{c},\ve{d}\}\cap fn(E)=\emptyset$)
\begin{eqnarray}
Q \wt{(\ve{d})\overline{a}B} Q' \mbox{ for some $\ve{d},B,Q'$} \label{output-goal-1}\\
\mbox{and for every $E[X]$, }
(\ve{c})(E[A]\para P') \approx \; \mathcal{R}\; \approx (\ve{d})(E[B]\para Q') \label{output-goal-2}
\end{eqnarray}
The argument can be conducted in a pretty similar way to that in the previous case for input; this time one attaches a process $E[Tr_m]$ and also uses the congruence properties of $\cong$. 

Since $P\cong Q$, we have $Q \wt{(\ve{d})\overline{a}B} Q'$ which fulfills (\ref{output-goal-1}), and also
\begin{equation}\label{mid-2}
(\ve{c})(P'\para !m(Z).A\lrangle{Z}) \cong  (\ve{d})(Q'\para  !m(Z).B\lrangle{Z})
\end{equation}
We know from the congruence properties and (\ref{mid-2}) that 
\begin{equation}\label{mid-3}
(m)((\ve{c})(P'\para !m(Z).A\lrangle{Z})\para E[Tr_m]) \cong  (m)((\ve{d})(Q'\para  !m(Z).B\lrangle{Z}) \para E[Tr_m])
\end{equation}
By some simple structural adjustment and the factorization theorem (Theorem \ref{factor-bigd}), we know the lhs (left hand side)  and rhs (right hand side) of (\ref{mid-3}) are equivalent to $(\ve{c})(E[A]\para P')$ and $(\ve{d})(E[B]\para Q')$ respectively. 
So we have
\[(\ve{c})(E[A]\para P') \approx \mbox{(lhs of (\ref{mid-3}))} \,\mathcal{R}\,  \mbox{(rhs of (\ref{mid-3}))} \approx (\ve{d})(E[B]\para Q')
\] This is exactly what  (\ref{output-goal-2}) says.

\end{itemize}
The proof is completed now.
\end{proof} 



\noindent\textbf{\emph{Remark}}.
\begin{itemize}
\item The success of the characterization of context bisimulation using normal bisimulation in $\Pi^D_n$ can be attributed to the fact that $\Pi^D_n$ processes are purely higher-order, i.e. no names can be passed but only (parameterized) processes (carrying names), and moreover no names can be parameterized. Thus one can delay the instantiation of a process parameterization by moving it elsewhere with the help of triggers. 

\item We mentioned in the introduction that in \cite{Cao06}, the index technique is utilized to show that strong normal bisimulation characterizes strong context bisimulation, within a calculus capable of both name-passing and process-passing but without any parameterization. 
A critical point there is that indices can be used to precisely pinpoint the matching of actions from the processes when going through some intermediate transformation (e.g. factorization). 
We believe, by combining the technique in that paper with the approach in this section, one can further show that $\sim$ and $\simeq$ coincide in the calculus $\Pi^D_n$. The key point in this combination is that we can reuse the approach in this section for proving the coincidence in the strong case, except that we need to harness indices to mark and filter out the extra $\tau$ actions brought about by the operation of factorization, and thus the precise matching of strong actions can be established. 


\end{itemize}





\section{Normal bisimulation in $\Pi^d_n$}\label{smalld-charac}
In this section, we examine context bisimulation in $\Pi^d_n$, particularly $\Pi^d_1$ for simplicity.
We show that, unlike that in $\Pi^D_1$, the technique of normal bisimulation \cite{San92}\cite{San94}  cannot be extended to $\Pi^d_1$ that is endowed with name parameterization instead of process parameterization. 
To this end, we show the negative fact in two steps. 
Firstly, we discuss the possible form of normal bisimulation, toward giving some intuition. Then we provide a counterexample, to show that the expected form of normal bisimulation does not work out. 
Therefore finding a useful characterization of context bisimulation may amount to exploiting further the essence (e.g. expressiveness) of $\Pi^d_n$.

\subsection*{Possible form of normal bisimulation}
Along the line of the very original idea of normal bisimulation \cite{San92}\cite{San94}, we pretend having the `normal bisimulation', among which the largest one is denoted by $\cong'$.  This would lead to the argument below.

\begin{itemize}

 \item A trigger now should be defined as $Tr_m\DEF \lrangle{z}\overline{m}z$, because it is supposed to carry names rather than processes. This immediately brings about a critical problem. That is, name-passing is not allowed in $\Pi^d_1$, and we only admit abstraction-passing. 

\item In the definition of $\cong'$, the output clause should take the following form. 
\begin{flushleft}
If $P \st{(\ve{c})\overline{a}A} P'$ then $Q \wt{(\ve{d})\overline{a}B} Q'$ for some $\ve{d},B,Q'$, 
and it holds that ($m$ is fresh)
\[(\ve{c})(P'\para E[A]) \; \mathcal{R}\;  (\ve{d})(Q'\para  E[B])
\]  
\end{flushleft}
where particularly it should be that $E[X]\DEF !m(z).X\lrangle{z}$ in line with the trigger form. 
Again, the special environment $E[X]$ does not belong to the calculus $\Pi^d_1$. 




\item As for the input clause of $\cong'$, one meets with similar obstacle.

\end{itemize}

So intuitively, the failure of trigger technique, and thus the failure of the factorization theorem, deprives $\Pi^d_1$ of the normal bisimulation. Furthermore, below we provide a counterexample to exhibit the deprival of normal bisimulation in $\Pi^d_1$.

\subsection*{A counterexample}
We examine the following example:
$
W\DEF A\lrangle{d}, \mbox{ in which } A\DEF \lrangle{x}\overline{x}
$. 
Obviously $W$ belongs to $\Pi^d_1$ and is able to fire an action on $d$, i.e.
$
W\equiv \overline{d} \,\st{\overline{d}}\, 0
$. 
However if one tries to factorize out the subprocess $A$, some contradiction arises by the examining below. 
\begin{enumerate}
\item[1)] One is supposed to replace $A$ with a trigger of certain (general) form, say $T$, which has to be an abstraction \emph{on a name} (to remain well-typed); so $T$ must take the shape $\lrangle{z}T'$, and we have after a substitution
\begin{equation}\label{smalld-counterex-1}
W\hosub{T}{A}\equiv T'\fosub{d}{z}
\end{equation}  Now some sugar should be added, i.e. some context $F$ is needed to contain $W\hosub{T}{A}$ so that the resulting process bi-simulates $W$. Let us suppose $F$ is of the form 
\[
(\ve{c})([\cdot] \para G[A])
\] in which $G$ should have $A$, in conformance to the rationale of factorization. Then generally, in terms of bisimulation, $F[T'\fosub{d}{z}]$ should engage in some internal moves between $T'\fosub{d}{z}$ and $G[A]$, so that in its current (different) position, $A$ can do the same action $\overline{d}$ after an instantiation on $x$. It is the responsibility of $T'$ to convey the particular information of $d$ to $G$.

\item[2)] Holding back a little bit, a crux is that there is no way to transmit, with the help of any process (let alone a trigger), a concrete name for instantiation of the abstraction (e.g. in (\ref{smalld-counterex-1})) to the newly-assigned place for future access, because all the processes here are strictly higher-order. 

\end{enumerate}


Hence we conclude that the technique of normal bisimulation does not work for $\Pi^d_n$. This result also sheds light on the expressiveness of $\Pi^d_n$. It reflects that name parameterization offers more flexibility than process parameterization in expressiveness so that its bisimulation does not have a similar (simpler) characterization. We give more discussion in the conclusion below.







\vspace*{-.3cm}
\section{Conclusion}\label{conclusion}
This paper provides some results on the characterization of context bisimulation in parameterized higher-order processes, and thus offers some tool as well as some insights for the bisimulation theory in higher-order paradigm. Firstly, we show that when extended with process parameterization, higher-order processes possess the characterization of context bisimulation by normal bisimulation. 
Secondly, we show that this technique of normal bisimulation, at least in its original form, cannot be extended to characterizing context bisimulation in presence of name parameterization. 
This separation result implies that the two kinds of parameterization have some essential difference. 
It will offer some potential reference for relevant research, for instance expressiveness studies.

There is some work worth further consideration.
\begin{itemize}
\item Finding some appropriate proof technique for the context bisimulation in $\Pi^d_n$.  To some extent, $\Pi^d_n$ is more useful than $\Pi^D_n$, which can strictly enhance the expressiveness of $\Pi$ \cite{LPSS10}, because it consists of some name-handling primitive, though it is still a higher-order language (i.e. no name-passing). So in order to make $\Pi^d_n$ more convenient when studying related topics, for example expressiveness and applications like modeling-verifying, there should be some proof technique for its canonical bisimulation equivalence, i.e. the context bisimulation. Since a standard normal characterization based on triggers is not available, one has to search for other ways to simplify the work on proving/checking whether two processes are context bisimilar. This may need some restriction on name abstraction or some other novel technique catering for certain motivation.   

\item Extension of the available characterization and related proof technique to more general higher-order models. For a long time, normal bisimulation and trigger technique has seen successful application and been deemed as a somewhat general method of dealing with higher-order processes (even applicable in some first-order process models). However, the result here implies that the normal approach, i.e. proof of context bisimulation up-to trigger and context, deserves more examination. Sometimes it would be better to study a general form of bisimulation in some general higher-order model. This is advantageous in that the usefulness of the available technique can be further tested, and also more essence of the behavior equivalence in higher-order processes can be hopefully revealed. A possible direction is to follow the idea in \cite{SKS11} where environmental bisimulation is proposed for higher-order models, and examine the proof technique of environmental bisimulation that is shown to coincide with canonical bisimulation in various higher-order models. In such a more general model, one may exploit some general proof technique (e.g. up-to certain normal feature) of different nature.
\end{itemize}







\noindent\textbf{Acknowledgement}\;\;
The author thanks Davide Sangiorgi for many constructive comments, and Qiang Yin for much helpful discussion.
He is also grateful to the comments and suggestions from the anonymous referees.  

\bibliographystyle{eptcs}
\bibliography{process}

\begin{thebibliography}{10}
\providecommand{\bibitemdeclare}[2]{}
\providecommand{\surnamestart}{}
\providecommand{\surnameend}{}
\providecommand{\urlprefix}{Available at }
\providecommand{\url}[1]{\texttt{#1}}
\providecommand{\href}[2]{\texttt{#2}}
\providecommand{\urlalt}[2]{\href{#1}{#2}}
\providecommand{\doi}[1]{doi:\urlalt{http://dx.doi.org/#1}{#1}}
\providecommand{\bibinfo}[2]{#2}

\bibitemdeclare{inproceedings}{Cao06}
\bibitem{Cao06}
\bibinfo{author}{Z.~\surnamestart Cao\surnameend} (\bibinfo{year}{2006}):
  \emph{\bibinfo{title}{More on Bisimulations for Higher Order
  $\pi$-Calculus}}.
\newblock In \bibinfo{editor}{L.~\surnamestart Aceto\surnameend} \&
  \bibinfo{editor}{A.~\surnamestart In{\'g}olfsd{\'o}ttir\surnameend}, editors:
  {\sl \bibinfo{booktitle}{Proceedings of FOSSACS2006}}, {\sl
  \bibinfo{series}{LNCS}} \bibinfo{volume}{3921}, pp. \bibinfo{pages}{63--78},
  \doi{10.1007/11690634\_5}.
\newblock \bibinfo{note}{Held as Part of the Joint European Conferences on
  Theory and Practice of Software, ETAPS 2006, Vienna, Austria, March 25-31,
  2006.}

\bibitemdeclare{article}{JR05}
\bibitem{JR05}
\bibinfo{author}{A.~\surnamestart Jeffrey\surnameend} \&
  \bibinfo{author}{J.~\surnamestart Rathke\surnameend} (\bibinfo{year}{2005}):
  \emph{\bibinfo{title}{Contextual equivalence for higher-order pi-calculus
  revisited}}.
\newblock {\sl \bibinfo{journal}{Logical Methods in Computer Science}}
  \bibinfo{volume}{1(1:4)}, \doi{10.2168/LMCS-1(1:4)2005}.

\bibitemdeclare{inproceedings}{LPSS10}
\bibitem{LPSS10}
\bibinfo{author}{I.~\surnamestart Lanese\surnameend}, \bibinfo{author}{J.~A.
  \surnamestart P\'{e}rez\surnameend}, \bibinfo{author}{D.~\surnamestart
  Sangiorgi\surnameend} \& \bibinfo{author}{A.~\surnamestart
  Schmitt\surnameend} (\bibinfo{year}{2010}): \emph{\bibinfo{title}{On the
  Expressiveness of Polyadic and Synchronous Communication in Higher-Order
  Process Calculi}}.
\newblock In: {\sl \bibinfo{booktitle}{Proceedings of ICALP 2010}},
  \bibinfo{series}{LNCS}, \bibinfo{publisher}{Springer Verlag}, pp.
  \bibinfo{pages}{442--453}, \doi{10.1007/978-3-642-14162-1\_37}.

\bibitemdeclare{article}{LPSS10a}
\bibitem{LPSS10a}
\bibinfo{author}{I.~\surnamestart Lanese\surnameend}, \bibinfo{author}{J.~A.
  \surnamestart P\'{e}rez\surnameend}, \bibinfo{author}{D.~\surnamestart
  Sangiorgi\surnameend} \& \bibinfo{author}{A.~\surnamestart
  Schmitt\surnameend} (\bibinfo{year}{2011}): \emph{\bibinfo{title}{On the
  Expressiveness and Decidability of Higher-Order Process Calculi}}.
\newblock {\sl \bibinfo{journal}{Information and Computation}}
  \bibinfo{volume}{209(2)}, pp. \bibinfo{pages}{198--226},
  \doi{10.1016/j.ic.2010.10.001}.

\bibitemdeclare{inproceedings}{LPSS08}
\bibitem{LPSS08}
\bibinfo{author}{I.~\surnamestart Lanese\surnameend}, \bibinfo{author}{J.A.
  \surnamestart Perez\surnameend}, \bibinfo{author}{D.~\surnamestart
  Sangiorgi\surnameend} \& \bibinfo{author}{A.~\surnamestart
  Schmitt\surnameend} (\bibinfo{year}{2008}): \emph{\bibinfo{title}{On the
  Expressiveness and Decidability of Higher-Order Process Calculi}}.
\newblock In: {\sl \bibinfo{booktitle}{Proceedings of the 23rd Annual IEEE
  Symposium on Logic in Computer Science (LICS 2008)}},
  \bibinfo{publisher}{IEEE Computer Society}, pp. \bibinfo{pages}{145--155},
  \doi{10.1109/LICS.2008.8}.
\newblock \bibinfo{note}{Journal version in \cite{LPSS10a}}.

\bibitemdeclare{article}{Mil92}
\bibitem{Mil92}
\bibinfo{author}{R.~\surnamestart Milner\surnameend} (\bibinfo{year}{1992}):
  \emph{\bibinfo{title}{Functions as Processes}}.
\newblock {\sl \bibinfo{journal}{Mathematical Structures in Computer Science}}
  \bibinfo{volume}{2(2)}, pp. \bibinfo{pages}{119--141},
  \doi{10.1017/S0960129500001407}.
\newblock \bibinfo{note}{Research Report 1154, INRIA, Sofia Antipolis, 1990}.

\bibitemdeclare{phdthesis}{San92}
\bibitem{San92}
\bibinfo{author}{D.~\surnamestart Sangiorgi\surnameend} (\bibinfo{year}{1992}):
  \emph{\bibinfo{title}{Expressing Mobility in Process Algebras: First-order
  and Higher-order Paradigms}}.
\newblock \bibinfo{type}{Phd thesis}, \bibinfo{school}{University of
  Edinburgh}.

\bibitemdeclare{inproceedings}{San92a}
\bibitem{San92a}
\bibinfo{author}{D.~\surnamestart Sangiorgi\surnameend} (\bibinfo{year}{1992}):
  \emph{\bibinfo{title}{From $\pi$-Calculus to Higher-Order
  $\pi$-Calculus---and Back}}.
\newblock In: {\sl \bibinfo{booktitle}{Proceedings of TAPSOFT '93}}, {\sl
  \bibinfo{series}{LNCS}} \bibinfo{volume}{668}, \bibinfo{publisher}{Springer
  Verlag}, pp. \bibinfo{pages}{151--166}, \doi{10.1007/3-540-56610-4\_62}.

\bibitemdeclare{article}{San94}
\bibitem{San94}
\bibinfo{author}{D.~\surnamestart Sangiorgi\surnameend} (\bibinfo{year}{1996}):
  \emph{\bibinfo{title}{Bisimulation for Higher-order Process Calculi}}.
\newblock {\sl \bibinfo{journal}{Information and Computation}}
  \bibinfo{volume}{131(2)}, pp. \bibinfo{pages}{141--178},
  \doi{10.1006/inco.1996.0096}.
\newblock \bibinfo{note}{Preliminary version in proceedings of PROCOMET'94
  (IFIP Working Conference on Programming Concepts, Methods and Calculi), pages
  207-224, North Holland, 1994}.

\bibitemdeclare{article}{San96}
\bibitem{San96}
\bibinfo{author}{D.~\surnamestart Sangiorgi\surnameend} (\bibinfo{year}{1996}):
  \emph{\bibinfo{title}{Pi-calculus, Internal Mobility and Agent-Passing
  Calculi}}.
\newblock {\sl \bibinfo{journal}{Theoretical Computer Science}}
  \bibinfo{volume}{167(2)}, \doi{10.1016/0304-3975(96)00075-8}.
\newblock \bibinfo{note}{Extracts of parts of the material contained in this
  paper can be found in the Proceedings of TAPSOFT'95 and ICALP'95}.

\bibitemdeclare{article}{San98}
\bibitem{San98}
\bibinfo{author}{D.~\surnamestart Sangiorgi\surnameend} (\bibinfo{year}{1998}):
  \emph{\bibinfo{title}{On the Bisimulation Proof Method}}.
\newblock {\sl \bibinfo{journal}{Mathematical Structures in Computer Science}}
  \bibinfo{volume}{8(6)}, pp. \bibinfo{pages}{447--479},
  \doi{10.1017/S0960129598002527}.
\newblock \bibinfo{note}{An extended abstract in Proceedings of MFCS'95, LNCS
  969, pp. 479-488, Springer Verlag}.

\bibitemdeclare{article}{SKS11}
\bibitem{SKS11}
\bibinfo{author}{D.~\surnamestart Sangiorgi\surnameend},
  \bibinfo{author}{N.~\surnamestart Kobayashi\surnameend} \&
  \bibinfo{author}{E.~\surnamestart Sumii\surnameend} (\bibinfo{year}{2011}):
  \emph{\bibinfo{title}{Environmental bisimulations for higher-order
  languages}}.
\newblock {\sl \bibinfo{journal}{ACM Transactions on Programming Languages and
  Systems}} \bibinfo{volume}{33(1)}, p.~\bibinfo{pages}{5},
  \doi{10.1145/1889997.1890002}.

\bibitemdeclare{book}{SW01a}
\bibitem{SW01a}
\bibinfo{author}{D.~\surnamestart Sangiorgi\surnameend} \&
  \bibinfo{author}{D.~\surnamestart Walker\surnameend} (\bibinfo{year}{2001}):
  \emph{\bibinfo{title}{The Pi-calculus: a Theory of Mobile Processes}}.
\newblock \bibinfo{publisher}{Cambridge Universtity Press}.

\bibitemdeclare{phdthesis}{Tho90}
\bibitem{Tho90}
\bibinfo{author}{B.~\surnamestart Thomsen\surnameend} (\bibinfo{year}{1990}):
  \emph{\bibinfo{title}{Calculi for Higher Order Communicating Systems}}.
\newblock \bibinfo{type}{Phd thesis}, \bibinfo{school}{Department of Computing,
  Imperial College}.

\bibitemdeclare{article}{Tho93}
\bibitem{Tho93}
\bibinfo{author}{B.~\surnamestart Thomsen\surnameend} (\bibinfo{year}{1993}):
  \emph{\bibinfo{title}{Plain {CHOCS}, a Second Generation Calculus for
  Higher-Order Processes}}.
\newblock {\sl \bibinfo{journal}{Acta Informatica}} \bibinfo{volume}{30(1)},
  pp. \bibinfo{pages}{1--59}, \doi{10.1007/BF01200262}.

\bibitemdeclare{inproceedings}{VD98}
\bibitem{VD98}
\bibinfo{author}{J.-L. \surnamestart Vivas\surnameend} \&
  \bibinfo{author}{M.~\surnamestart Dam\surnameend} (\bibinfo{year}{1998}):
  \emph{\bibinfo{title}{From Higher-Order Pi-Calculus to Pi-Calculus in the
  Presence of Static Operators}}.
\newblock In: {\sl \bibinfo{booktitle}{Proceedings of the 9th International
  Conference on Concurrency Theory}}, {\sl \bibinfo{series}{LNCS}}
  \bibinfo{volume}{1466}, pp. \bibinfo{pages}{115--130},
  \doi{10.1007/BFb0055619}.
\newblock \bibinfo{note}{Nice, France, September 8-11, 1998}.

\bibitemdeclare{article}{Xu12}
\bibitem{Xu12}
\bibinfo{author}{Xian \surnamestart Xu\surnameend} (\bibinfo{year}{2012}):
  \emph{\bibinfo{title}{Distinguishing and Relating Higher-order and
  First-order Processes by Expressiveness}}.
\newblock {\sl \bibinfo{journal}{Acta Informatica}} \bibinfo{volume}{49(7-8)},
  pp. \bibinfo{pages}{445--484}, \doi{10.1007/s00236-012-0168-9}.

\end{thebibliography}

\clearpage

\appendix
\noindent\textbf{\Large Appendix}


\section{Proof of Section \ref{bigD-charac}}\label{appendix-proof-bigd-factor}
In this appendix, we give the proof of the factorization theorem in section \ref{bigD-charac}.
\begin{proof}[Proof of Theorem \ref{factor-bigd}: Factorization]
The proof is by induction on the structure of $E$. The cases 1,2,3 are the base cases. 
\begin{enumerate}
\item $E$ is $0$ or $E$ is $Y$ and $Y\neq X$. These cases are trivial. 

\item 
$E$ is $X$.  
Notice this time $E[Tr_m]$ is $Tr_m\equiv \lrangle{Z}\overline{m}Z\equiv \lrangle{Y}\overline{m}Y$ (up-to alpha-conversion), so  the second statement of this theorem applies and the two terms of interest are 
\[
A \quad\mbox{ and }\quad \lrangle{Y}((m)(\overline{m}Y\para !m(Z).A\lrangle{Z})
\]
Suppose $A\equiv \lrangle{Z'}F$, then it is easy to verify that for each $B$ one has
\[
F\lrangle{B} \approx (m)(\overline{m}B\para !m(Z).F\hosub{Z}{Z'}
\] which completes this case. 


\item 
$E$ is $X\lrangle{E_1}$. 
This can be proven by showing the following relation $\mathcal{R}$ is a context bisimulation up-to $\sim$ (this technique is standard \cite{San94}\cite{SW01a} and we omit its definition here).  
\[\mathcal{R} \DEF \{(A\lrangle{E_1}, (m)(Tr_m\lrangle{E_1}\para !m(Z).A\lrangle{Z}) \,|\, m \mbox{ is fresh }\} \,\cup \approx
\] Let $(A\lrangle{E_1}, (m)(Tr_m\lrangle{E_1}\para !m(Z).A\lrangle{Z}) \in \mathcal{R}$ and $A\equiv \lrangle{Z'}T$; so the pair is actually $$(T\hosub{E_1}{Z'},  (m)(\overline{m}E_1 \para !m(Z).T\hosub{Z}{Z'})$$
There are mainly two cases to consider. 
\begin{itemize}
\item $(m)(\overline{m}E_1 \para !m(Z).T\hosub{Z}{Z'}) \st{\alpha} T'$. Then $\alpha$ must be $\tau$, and thus  
$$T'\equiv (m)(T\hosub{E_1}{Z'}\para !m(Z).T\hosub{Z}{Z'})$$ 
By $T\hosub{E_1}{Z'}\wt{} T\hosub{E_1}{Z'}$ (null transition),  since $m$ is fresh, it can be easily seen that, 
\[
T\hosub{E_1}{Z'} \sim T\hosub{E_1}{Z'} \,\mathcal{R}\, T\hosub{E_1}{Z'}  \sim T'
\]

\item $T\hosub{E_1}{Z'} \st{\alpha} T_1$. Then this is simulated by
\[
\begin{array}{ll}
 & (m)(\overline{m}E_1 \para !m(Z).T\hosub{Z}{Z'}) \\
\st{\tau} & (m)(T\hosub{E_1}{Z'}\para !m(Z).T\hosub{Z}{Z'}) \\
\st{\alpha} & (m)(T_1\para !m(Z).T\hosub{Z}{Z'}) \DEF T_2
\end{array}
\] So it holds in a straightforward way that
\[
T_1 \sim T_1 \,\mathcal{R}\, T_1 \sim T_2
\]
\end{itemize}

\noindent The following cases 4,5,6,7,8,9 are cases involving the induction hypothesis (ind. hyp. for short). Notice that we will only consider the case when statement (1) in the theorem applies, and the case for (2) can be dealt with similarly.

\item 
$E$ is $\lrangle{Y}E_1$.
Then we have by ind. hyp.
\[
\begin{array}{ll}
& E[A] \equiv \lrangle{Y}E_1[A] \\
\approx & \lrangle{Y}((m)(E_1[Tr_m] \para  !m(Z).A\lrangle{Z}))
\end{array}
\] To conclude this case, we have to show
\[\lrangle{Y}((m)(E_1[Tr_m] \para  !m(Z).A\lrangle{Z})) \approx  \lrangle{Y}((m)(E_1[Tr_m] \para  !m(Z).A\lrangle{Z}))
\]
This is immediate and the right-hand-side is exactly the second claim (2) of this theorem.

\item 
$E$ is $\overline{a}E_2.E_1$. 
Then we have
\[
\begin{array}{llr}
& E[A] \equiv \overline{a}E_2[A].E_1[A] & \\
\approx & \overline{a}[(m)(E_2[Tr_m] \para  !m(Z).A\lrangle{Z})].((m)(E_1[Tr_m] \para  !m(Z).A\lrangle{Z})) & (\mbox{ind. hyp.}) \\
\approx & (m)(\overline{a}[(m)(E_2[Tr_m] \para  !m(Z).A\lrangle{Z})].E_1[Tr_m] \para  !m(Z).A\lrangle{Z}) & (\mbox{Lemma \ref{premise-factor-0}(1)}) \\
\approx & (m)(\overline{a}E_2[Tr_m].E_1[Tr_m] \para  !m(Z).A\lrangle{Z}) & (\mbox{Lemma \ref{premise-factor-0}(2)}) \\
\end{array}
\] The last equation gives exactly what we expect.


\item 
 $E$ is $a(Y).E_1$. 
This is similar (and easier) than the previous case. 

\item 
$E$ is $E_1\para E_2$. 
Then we have 
\[
\begin{array}{llr}
& E[A] \equiv E_1[A]\para E_2[A] & \\
\approx & (m)(E_1[Tr_m] \para  !m(Z).A\lrangle{Z}) \para (m)(E_2[Tr_m] \para  !m(Z).A\lrangle{Z}) & (\mbox{ind. hyp.}) \\
\approx & (m)(E_1[Tr_m] \para E_2[Tr_m] \para  !m(Z).A\lrangle{Z}) &  (\mbox{Lemma \ref{premise-factor-0}(3)})
\end{array}
\] The last equation completes this case.


\item 
$E$ is $ (c)E_1$. 
This is similar to the previous case.


\item 
 $E$ is $E_2\lrangle{E_1}$. 
This case can be reduced to the case $E$ is $Y\lrangle{E_1}$ and $Y\neq X$, because if $E_2$ is not a variable (i.e. an abstraction) or is $X$, then it can be handled in a way that falls into one of the previous cases (up-to structural congruence).  So we have 
\[
\begin{array}{ll}
& E[A] \\
\equiv & Y\lrangle{E_1[A]} \\
\approx & Y\lrangle{(m)(E[Tr_m] \para  !m(Z).A\lrangle{Z})} \DEF T \qquad (\mbox{ind. hyp.}) 
\end{array}
\]
We want to show that 
$T \approx (m)(Y\lrangle{E_1[Tr_m]} \para  !m(Z).A\lrangle{Z})
$ i.e. for every $B$, 
\[B\lrangle{(m)(E_1[Tr_m] \para  !m(Z).A\lrangle{Z})} \approx (m)(B\lrangle{E_1[Tr_m]} \para  !m(Z).A\lrangle{Z}) 
\] This is immediately from Lemma \ref{premise-factor-1}. 
\end{enumerate}

The proof is now completed.
\end{proof}



\end{document}